\documentclass[11pt]{article}

\usepackage{graphicx, subfigure}
\usepackage[top=.75in,bottom=.75in,left=.75in,right=.75in]{geometry}
\usepackage{amsmath,amsfonts,amscd,amssymb,bm,epsfig,epsf,bbm,constants}
\usepackage{algorithm,algorithmic}
\usepackage[usenames]{color}

\definecolor{darkred}{RGB}{100,0,0}
\definecolor{darkgreen}{RGB}{0,100,0}
\definecolor{darkblue}{RGB}{0,0,150}
\definecolor{red}{RGB}{255,0,0}

\usepackage{hyperref}
\hypersetup{colorlinks=true, linkcolor=darkred, citecolor=darkgreen, urlcolor=darkblue}
\usepackage{url}

\newtheorem{theorem}{Theorem}[section]
\newtheorem{lemma}[theorem]{Lemma}
\newtheorem{corollary}[theorem]{Corollary}

\newcommand{\R}{\mathbb{R}}
\renewcommand{\C}{\mathbb{C}}

\newcommand{\<}{\langle}
\renewcommand{\>}{\rangle}

\newcommand{\goto}{\rightarrow}

\renewcommand{\P}{\operatorname{\mathbb{P}}}
\newcommand{\E}{\operatorname{\mathbb{E}}}

\newcommand{\cA}{\mathcal{A}}
\newcommand{\cN}{\mathcal{N}}

\newcommand{\vct}[1]{\bm{#1}}
\newcommand{\mtx}[1]{\bm{#1}}

\newcommand{\lspan}[1]{\operatorname{span}{#1}}

\newcommand{\rank}{\operatorname{rank}}

\newcommand{\tr}{\operatorname{tr}}

\numberwithin{equation}{section}

\definecolor{eac}{RGB}{200,50,50}
\definecolor{mad}{RGB}{50,200,50}
\definecolor{ejc}{RGB}{50,50,200}

\numberwithin{equation}{section}
\def \endprf{\hfill {\vrule height6pt width6pt depth0pt}\medskip}
\newenvironment{proof}{\noindent {\bf Proof} }{\endprf\par}

\begin{document}

\title{Solving Quadratic Equations via PhaseLift\\
  when There Are About As Many Equations As Unknowns} \author{
  Emmanuel J.~Cand\`es%
  \footnote{Departments of Mathematics and Statistics, Stanford
    University, Stanford CA 94305 
} \
and Xiaodong Li \footnote{Department of Mathematics, Stanford
  University, Stanford CA 94305
}
}

\date{August 2012}
\maketitle

\begin{abstract}
  This note shows that we can recover any complex vector $\vct{x}_0
  \in \C^n$ exactly from on the order of $n$ quadratic equations of
  the form $|\<\vct{a}_i, \vct{x}_0\>|^2 = b_i$, $i = 1, \ldots, m$,
  by using a semidefinite program known as PhaseLift. This improves
  upon earlier bounds in \cite{PhaseLift2}, which required the number
  of equations to be at least on the order of $n \log n$. Further, we
  show that exact recovery holds for all input vectors simultaneously,
  and also demonstrate optimal recovery results from noisy quadratic
  measurements; these results are much sharper than previously known
  results.
\end{abstract}

\medskip


\section{Introduction}
\label{sec:intro}

Suppose we wish to solve quadratic equations of the form 
\begin{equation}
  \label{eq:quadratic}
  |\<\vct{a}_i, \vct{x}_0\>|^2 = b_i, \quad i = 1, \ldots, m,  
\end{equation}
where $\vct{x}_0 \in \mathbb{C}^n$ is unknown and $\vct{a}_i \in
\mathbb{C}^n$ and $b_i \in \mathbb{R}$ are given. This is a fundamental
problem which includes all phase retrieval problems in which one
cannot measure the phase of the linear measurements $\<\vct{a}_i,
\vct{x}\>$, only their magnitude. Recently,
\cite{PhaseLift1,PhaseLift2} proposed finding solutions to
\eqref{eq:quadratic} by convex programming techniques. The idea can be
explained rather simply: lift the problem in higher dimensions and
write $\mtx{X} = \vct{x}\vct{x}^*$ so that \eqref{eq:quadratic} can be
formulated as
\begin{equation*}
\label{eq:NP}
  \begin{array}{ll}
    \text{find}   & \quad \mtx{X}\\
    \text{subject to} & \quad  \mtx{X} \succeq 0, \, \, \rank(\mtx{X}) = 1,\\
& \quad \tr(\vct{a}_i \vct{a}_i^* \mtx{X}) = b_i, \, \, i = 1, \ldots m. 
\end{array}
\end{equation*}
Then approximate this combinatorially hard problem by using a convex
surrogate for the nonconvex rank functional: {\em PhaseLift}
\cite{PhaseLift1,PhaseLift2} is the relaxation
\begin{equation}
  \label{eq:PL}
 \begin{array}{ll}
   \text{minimize}   & \quad \tr(\mtx{X})\\
   \text{subject to} & \quad  \mtx{X} \succeq 0,\\
   & \quad \tr(\vct{a}_i \vct{a}_i^* \mtx{X}) = b_i, \, \, i = 1, \ldots, m. 
\end{array}
\end{equation}
The main result in \cite{PhaseLift2} states that if the equations are
sufficiently randomized and their number $m$ is at least on the order
of $n \log n$, then the solution to the convex relaxation
\eqref{eq:PL} is exact.
\begin{theorem}[\cite{PhaseLift2}]
\label{teo:PL}
Fix $\vct{x}_0 \in \mathbb{C}^n$ arbitrarily and 
suppose that
\begin{equation}
\label{eq:PL}
m \ge c_0 \, n \log n,
\end{equation} 
where $c_0$ is a sufficiently large constant. Then in all models
introduced below, PhaseLift is exact with probability at least $1 -
3e^{-\gamma \frac{m}{n}}$ ($\gamma$ is a positive numerical constant)
in the sense that \eqref{eq:PL} has a unique solution equal to
$\vct{x}_0 \vct{x}_0^*$.\footnote{Upon retrieving $\hat{\mtx{X}} =
  \vct{x}_0 \vct{x}_0^*$, a simple factorization recovers $\vct{x}_0$
  up to global phase, i.e.~multiplication by a complex scalar of unit
  magnitude.} 
\end{theorem}
The models above are either complex or real depending upon whether
$\vct{x}_0$ is complex or real valued. In all cases the $\vct{a}_i$'s
are independently and identically distributed with the following
distributions:
\begin{itemize}
\item {\em Complex models.} The uniform distribution on the complex
  sphere of radius $\sqrt{n}$, or the complex normal distribution
  $\mathcal{N}(0, \mtx{I_n}/2) + i\mathcal{N}(0, \mtx{I_n}/2)$.
\item {\em Real models.} The uniform distribution on the sphere of
  radius $\sqrt{n}$, or the normal distribution $\mathcal{N}(0,
  \mtx{I_n})$.
\end{itemize}

Clearly, one needs at least on the order of $n$ equations to have a
well posed problem, namely, a unique solution to
\eqref{eq:quadratic}.\footnote{The work in \cite{Edidin} shows that
  with probability one, $m = 4n-2$ randomized equations as in Theorem
  \ref{teo:PL} are sufficient for the intractable phase retrieval
  problem \eqref{eq:quadratic} to have a unique solution.} This raises
natural questions: 
\begin{description}
\item {\em Does the convex relaxation \eqref{eq:PL} with a number of
    equations on the order of the number of unknowns succeed?  Or is
    the lower bound \eqref{eq:PL} sharp? 
  \item Is it possible to improve the guaranteed probability of
    success?
\item Can we hope for a universal result stating that once
  the vectors $\vct{a}_i$ have been selected, \underline{all} input signals
  $\vct{x}_0$ can be recovered?}
\end{description}
This paper answers these questions.
\begin{theorem}
\label{teo:main}
Consider the setup of Theorem \ref{teo:PL}. Then for all $\vct{x}_0$
in $\C^n$ or $\R^n$, the solution to PhaseLift is exact with
probability at least $1 - O(e^{-\gamma m})$ if the number of equations
obeys
  \begin{equation}
\label{eq:main}
m \ge c_0 \,  n,
\end{equation}
where $c_0$ is a sufficiently large constant. Thus, exact recovery
holds simultaneously over all input signals.
\end{theorem}
In words, (1) the solution to most systems of quadratic equations can
be obtained by semidefinite programming as long as the number of
equations is at least a constant times the number of unknowns; (2) the
probability of failure is {\em exponentially} small in the number of
measurements, a significant sharpening of Theorem \ref{teo:PL}; (3)
these properties hold universally as explained above.

Letting $\cA: \C^{n \times n} \goto \R^m$ be the linear map
$\cA(\mtx{X}) = \{\tr(\vct{a}_i \vct{a}_i^* \mtx{X})\}_{1 \le i \le
  m}$, Theorem \ref{teo:PL} states that with high probability, the
null space of $\cA$ is tangent to the positive semidefinite (PSD) cone
$\{\mtx{X}: \mtx{X} \succeq \mtx{0}\}$ at a {\em fixed} $\vct{x}_0 \in
\C^n$. In constrast, Theorem \ref{teo:main} asserts that this
nullspace is tangent to the PSD cone at {\em all} rank-one
elements. Mathematically, what makes this possible is the sharpening
of the probability bounds; that is to say, the fact that for a fixed
$\vct{x}_0$, recovery holds with probability at least $1 -
O(e^{-\gamma m})$. Importantly, this improvement cannot be obtained
from the proof of Theorem \ref{teo:PL}. For instance, the argument in
\cite{PhaseLift2} does not allow removing the logarithmic factor in
the number of equations; consequently, although the general
organization of our proof is similar to that in \cite{PhaseLift2}, a
different argument is needed.

In most applications of interest, we do not have noiseless data but
rather observations of the form
\begin{equation}
  \label{eq:quadratic-noise}
  b_i = |\<\vct{a}_i, \vct{x}_0\>|^2 + w_i, \quad i = 1, \ldots, m,  
\end{equation}
where $w_i$ is a noise term. Here, we suggest recovering the signal by
solving
\begin{equation}
  \label{eq:PLn}
 \begin{array}{ll}
   \text{minimize}   & 
   \quad \sum_{1\le i \le m} |\tr(\vct{a}_i \vct{a}_i^* \mtx{X}) - b_i|\\
   \text{subject to} & \quad  \mtx{X} \succeq 0.
\end{array}
\end{equation}
The proposal cannot be simpler: find the positive semidefinite matrix
$\mtx{X}$ that best fits the observed data in an $\ell_1$ sense. One
can then extract the best-rank one approximation to recover the
signal. Our second result states that this procedure is accurate. 
\begin{theorem}
\label{teo:stability}
Consider the setup of Theorem \ref{teo:main}. Then for all $\vct{x}_0
\in \C^n$, the solution to \eqref{eq:PLn} obeys
\begin{equation}
  \label{eq:stability}
  \|\hat{\mtx{X}} - \vct{x}_0\vct{x}_0^*\|_F \le C_0 \frac{\|\vct{w}\|_1}{m}
\end{equation}
for some numerical constant $C_0$. For the Gaussian models, this holds
with the same probability as in the noiseless case whereas the
probability of failure is exponentially small in $n$ in the uniform
model. By finding the largest eigenvector with largest eigenvalue of
$\hat{\mtx{X}}$, one can also construct an estimate obeying
\begin{equation}
  \label{eq:stability2}
  \|\hat{\vct{x}} - e^{i\phi} \vct{x}_0\|_2 \le C_0 \min\Bigl(\|\vct{x_0}\|_2, \frac{\|\vct{w}\|_1}{m \|\vct{x}_0\|_2}\Bigr) 
\end{equation}
for some $\phi \in [0, 2\pi]$.
\end{theorem}
In Section \ref{sec:stability}, we shall explain that these results
are optimal and cannot possibly be improved. For now, we would like to
stress that the bounds \eqref{eq:stability}--\eqref{eq:stability2}
considerably strengthen those found in \cite{PhaseLift2}. To be sure,
this reference shows that if the noise $\vct{w}$ is known to be
bounded, i.e.~$\|\vct{w}\|_2 \le \varepsilon$, then a relaxed version
of \eqref{eq:PL} yields an estimate $\tilde{\mtx{X}}$ obeying
\[
\|\tilde{\mtx{X}} - \vct{x}_0\vct{x}_0^*\|_F^2 \le C_0 \, \varepsilon^2.
\]
In contrast, since $\|\vct{w}\|_1 \le \sqrt{m} \|\vct{w}\|_2 \le
\sqrt{m} \varepsilon$, the new Theorem \ref{teo:stability} gives
\[
\|\tilde{\mtx{X}} - \vct{x}_0\vct{x}_0^*\|_F^2 \le C_0 \,
\frac{\varepsilon^2}{m};
\]
this represents a substantial improvement.

\section{Proofs}
\label{sec:proofs}

We prove Theorems \ref{teo:main} and \ref{teo:stability} in the
real-valued case, the complex case being similar, see
\cite{PhaseLift2} for details. Next, the Gaussian and uniform models
are nearly equivalent: indeed, suppose $\vct{a}_i$ is uniformly
sampled on the sphere; if $n \rho_i^2 \sim \chi^2_n$ and is
independent of $\vct{a}_i$, then $\vct{z}_i = \rho_i \vct{a}_i$ is
normally distributed. Hence,
\[
b_i = |\<\vct{a}_i, \vct{x}_0\>|^2 + w_i \quad \Longleftrightarrow \quad
b'_i = |\<\vct{z}_i, \vct{x}_0\>|^2 + \rho_i^2 w_i.
\]
In the noiseless case, we have full equivalence. In the noisy case, we
can transfer a bound for Gaussian measurements into the same bound for
uniform measurements by changing the probability of success ever so
slightly---as noted in Theorem \ref{teo:stability}. Thus, it suffices
to study the real-valued Gaussian case.

We introduce some notation and with $[m] = \{1, \ldots, m\}$, we let
$\cA: \R^{n \times n} \goto \R^m$ be the linear map $\cA(\mtx{X}) =
\{\tr(\vct{a}_i \vct{a}_i^* \mtx{X})\}_{i \in [m]}$ whose adjoint is
given by $\cA^*(\vct{y}) = \sum_{i \in [m]} y_i \vct{a}_i
\vct{a}_i^*$. Note that vectors and matrices are boldfaced whereas
scalars are not. In the sequel, we let $T$ be the subspace of
symmetric matrices of the form $\{\mtx{X} = \vct{x} \vct{x}_0^* +
\vct{x}_0 \vct{x}^* : \vct{x} \in \R^n\}$ and $T^\perp$ be its
orthogonal complement. For a symmetric matrix $\mtx{X}$, we put
$\mtx{X}_T$ for the orthogonal projection of $\mtx{X}$ onto $T$ and
likewise for $\mtx{X}_{T^\perp}$. Hence, $\mtx{X} = \mtx{X}_T +
\mtx{X}_{T^\perp}$. Finally, $\|\vct{y}\|_p$ is the $\ell_p$ norm of a
vector $\vct{y}$ and $\|\mtx{X}\|$ (resp.~$\|\mtx{X}\|_F$) is the
spectral (resp.~Frobenius) norm of a matrix $\mtx{X}$.

\subsection{Dual certificates} 

We begin by specializing Lemmas 3.1 and 3.2 from \cite{PhaseLift2}.
\begin{lemma}[\cite{PhaseLift2}]
\label{lem:PL}
There is an event $E$ of probability at least $1 - 5e^{-\gamma_0 m}$
such that on $E$, any positive symmetric matrix obeys 
\[
m^{-1} \|\cA(\mtx{X})\|_1 \le (1 + 1/8) \tr(\mtx{X}),  
\]
and any symmetric rank-2 matrix obeys
\[
m^{-1} \|\cA(\mtx{X})\|_1 \ge 0.94 (1 - 1/8) \|\mtx{X}\|.
\]
\end{lemma}
The following intermediate result is novel, although we became aware
of a similar argument in \cite{DemanetHand} as we finished this paper. 
\begin{lemma}
\label{lem:inexact}
Suppose there is a matrix $\mtx{Y}$ in the range of $\cA^*$ obeying
$\mtx{Y}_{T^\perp}\preceq -\mtx{I}_{T^\perp}$ and $\|\mtx{Y}_T\|_F\leq
{1 \over {2}}$. Then on the event $E$ from Lemma \ref{lem:PL},
$\mtx{X}_0 = \vct{x}_0 \vct{x}_0^*$ is PhaseLift's unique feasible
point. 
\end{lemma}
\begin{proof}
  Suppose $\vct{x}_0 \vct{x}_0^* + \mtx{H}$ is feasible, which implies
  that (1) $\mtx{H}_{T^\perp} \succeq \mtx{0}$ and (2) $\mtx{H}$ is in
  the null space of $\cA$ so that $\<\mtx{H}, \mtx{Y}\> = 0 =
  \<\mtx{H}_T, \mtx{Y}_T\> + \<\mtx{H}_{T^\perp},
  \mtx{Y}_{T^\perp}\>$. On the one hand,
\[
\langle \mtx{H}_T, \mtx{Y}_T \rangle=-\langle \mtx{H}_{T^\perp},
\mtx{Y}_{T^\perp}\rangle\geq \langle \mtx{H}_{T^\perp},
\mtx{I}_{T^\perp}\rangle= \tr(\mtx{H}_{T^\perp}).
\]
Lemma \ref{lem:PL} asserts that $m^{-1} \|\cA(\mtx{H}_{T^\perp})\|_1
\le (1 + 1/8) \tr(\mtx{H}_{T^\perp})$ and $m^{-1}
\|\cA(\mtx{H}_{T})\|_1 \ge 0.94 (1 - 1/8) \|\mtx{H}_T\|$. Since
$\cA(\mtx{H}_{T}) = -\cA(\mtx{H}_{T^\perp})$, this gives
\begin{equation}
\label{eq:one}
\tr(\mtx{H}_{T^\perp})\geq {1\over
  {(1+1/8)m}}\|\cA(\mtx{H}_{T^\perp})\|_1 \geq 0.73
\|\mtx{H}_T\|\geq{0.73 \over {\sqrt{2}}}\|\mtx{H}_T\|_F, 
\end{equation}
where the last inequality is a consequence of the fact that $\mtx{H}_T$ has
rank at most 2.  On the other hand,
\begin{equation}
\label{eq:two}
|\langle \mtx{H}_T, \mtx{Y}_T \rangle|\leq
\|\mtx{H}_T\|_F\|\mtx{Y}_T\|_F\leq \frac12 \|\mtx{H}_{T}\|_F.
\end{equation}
Since $0.73/\sqrt{2} > 1/2$, \eqref{eq:one} and \eqref{eq:two} give
that $\mtx{H}_T = 0$, which in turns implies that
$\mtx{H}_{T^\perp}=\mtx{0}$. This completes the proof.
\end{proof}

To prove Theorem \ref{teo:main}, it remains to construct a matrix
$\mtx{Y}$ obeying the conditions of Lemma \ref{lem:inexact} for all
$\vct{x}_0 \in \R^n$. We proceed in two steps: we first show that for
a fixed $\vct{x}_0$, one can find $\mtx{Y}$ with high probability, and
then use this property to show that one can find $\mtx{Y}$ for all
$\vct{x}_0$.

\begin{lemma}
\label{lem:core}
Fix $\vct{x}_0 \in \R^n$. Then with probability at least $1 -
O(e^{-\gamma m})$, there exists $\mtx{Y}$ obeying $\|\mtx{Y}_{T^\perp} +{{17}\over {10}}\mtx{I}_{T^\perp}\|\leq{1\over {10}}$ and $\|\mtx{Y}_T\|_F\leq
{3 \over {20}}$. In addition, one can take $\mtx{Y} =
\cA^*(\vct{\lambda})$ with $\|\vct{\lambda}\|_\infty \le 7/m$.
\end{lemma}

\begin{proof}
  We assume that $\|\vct{x}_0\|_2 = 1$ without loss of generality.
  Our strategy is to show that
\begin{equation}
\label{eq:Y}
\mtx{Y} = \sum_{i \in [m]} \lambda_i  \vct{a}_i\vct{a}_i^T :=
\frac{1}{m} \sum_{i \in [m]} [|\< \vct{a}_i, \vct{x}_0\>|^2 \mathbbm{1}(|\< \vct{a}_i, \vct{x}_0\>|\leq 3) - \beta_0]  \, \vct{a}_i\vct{a}_i^T := \mtx{Y}^{(0)} - \mtx{Y}^{(1)},  
\end{equation}
where $\beta_0 = \E z^4 \mathbbm{1}(|z| \le 3) \approx 2.6728$ with $z
\sim \mathcal{N}(0,1)$, is a valid certificate. As claimed,
$\|\vct{\lambda}\|_\infty \le 7/m$.

We begin by checking the condition $\mtx{Y}_{T^\perp}\preceq
-\mtx{I}_{T^\perp}$. First, the matrix $\mtx{Y}^{(1)}$ is Wishart and
standard results in random matrix theory---e.g.~Corollary 5.35 in
\cite{VershyninRMT}---assert that
\begin{equation*}
\|\mtx{Y}^{(1)} - \E \mtx{Y}^{(1)} \| = \|\mtx{Y}^{(1)} - \beta_0 \mtx{I}\| \le \beta_0/40
\end{equation*}
with probability at least $1 - 2 e^{-\gamma m}$ provided that $m \ge c n$, where $c$ is sufficiently large. In particular, we have
\begin{equation}
\label{eq:YoneTp}
\|\mtx{Y}_{T^\perp}^{(1)} - \beta_0 \mtx{I}_{T^\perp}\| \le \beta_0/40.
\end{equation}
Second, letting $\vct{x}'$ be the projection of $\vct{x}$ onto the
orthogonal complement of $\lspan(\vct{x}_0)$, we have
\[
\mtx{Y}^{(0)}_{T^\perp} = \frac{1}{m} \sum_{i \in [m]} \vct{\xi}_i
\vct{\xi}_i^T, \quad \vct{\xi}_i = \< \vct{a}_i, \vct{x}_0\>  \mathbbm{1}(|\< \vct{a}_i, \vct{x}_0\>|\leq 3)
\, \vct{a}'_i.
\]
It is immediate to check that the $\vct{\xi}_i$'s are iid copies of a
zero-mean, isotropic and sub-Gaussian random vector $\vct{\xi}$. In
particular, with $z \sim \mathcal{N}(0,1)$,
\[
\E \vct{\xi} \vct{\xi}^T = \alpha_0 \, \mtx{I}_{T^\perp}, \quad
\alpha_0 = \E z^2 \mathbbm{1}(|z|\leq 3) \approx 0.9707.
\]
Again, standard results about random matrix with sub-gaussian
rows---e.g.~Theorem 5.39 in \cite{VershyninRMT}---give 
\begin{equation}
\label{eq:YzeroTp}
\|\mtx{Y}_{T^\perp}^{(0)} - \E \mtx{Y}_{T^\perp}^{(0)} \| =
\|\mtx{Y}_{T^\perp}^{(0)} - \alpha_0 \, \mtx{I}_{T^\perp}\| \le
\alpha_0/40
\end{equation}
with probability at least $1 - 2 e^{-\gamma m}$ provided that $m \ge c
n$, where $c$ is sufficiently large. Clearly, \eqref{eq:YoneTp} together with \eqref{eq:YzeroTp} yield the
first condition $\|\mtx{Y}_{T^\perp} +{{17}\over {10}}\mtx{I}_{T^\perp}\|\leq{1\over {10}}$ .

We now establish $\|\mtx{Y}_{T}\|_F \le 3/20$. To begin with, set
$\vct{y} = \mtx{Y} \vct{x}_0$ and observe that since
$\|\mtx{Y}_{T}\|^2_F = |\<\vct{y} , \vct{x}_0\>|^2 +
2\|\vct{y}'\|^2_2$, it suffices to verify that
\[
|\<\vct{y} , \vct{x}_0\>|^2 \le 1/20 \text{ and }  \|\vct{y}'\|^2_2 \le 1/10.
\]
Write $\<\vct{y} , \vct{x}_0\> = \frac{1}{m} \sum_{i \in [m]} \xi_i$,
where $\xi_i$ are iid copies of $\xi = z^4 \mathbbm{1}(|z| \le 3) -
\beta_0 z^2$, $z \sim \mathcal{N}(0,1)$. Note that $\xi$ is a
mean-zero sub-exponential random variable since the first term is
bounded and the second is a squared Gaussian variable. Thus,
Bernstein's inequality---e.g.~Corollary 5.17 in
\cite{VershyninRMT}---gives
\[
\P(|\<\vct{y} , \vct{x}_0\>| \ge 1/\sqrt{20}) \le 2 \exp(-\gamma m)
\]
for some numerical constant $\gamma$. Finally, write $\vct{y}'$ as 
\[
\vct{y}' = \frac{1}{m} \mtx{Z'} \vct{c}, \quad \mtx{Z}' = [\vct{a}'_1,
\ldots, \vct{a}'_m], \quad c_i = \<\vct{a}_i , \vct{x}_0\>^3
\mathbbm{1}(|\<\vct{a}_i, \vct{x}_0\>| \le 3) - \beta_0 \<\vct{a}_i,
\vct{x}_0\>, \, i \in [m].
\]
Note that $\mtx{Z}'$ and $\vct{c}$ are independent. On the one hand,
the $c_i^2$'s are iid sub-exponential variables and Corollary 5.17 in
\cite{VershyninRMT}---gives
\[
\P(\|\vct{c}\|_2^2 - \E \|\vct{c}\|_2^2 \ge m) \le 2 e^{-\gamma m} 
\]
for some numerical constant $\gamma > 0$. This shows that 
\begin{equation}
\label{eq:bernstein}
\|\vct{c}\|_2^2 \le (\delta_0 + 1)m, \quad 
\delta_0 = \E (z^3 \mathbbm{1}(|z| \le 3)
- \beta_0 z)^2 \approx 4.0663, \, z \sim \mathcal{N}(0,1), 
\end{equation}
with probability at least $1 - 2 e^{-\gamma m}$.  On the other hand,
for a fixed $\vct{x}$ obeying $\|\vct{x}\|_2 = 1$, $\|\mtx{Z'}
\vct{x}\|_2^2$ is distributed as a $\chi^2$-random variable with $n-1$
degrees of freedom and it follows that
\begin{equation}
  \label{eq:chi-sq}
  \P(\|\mtx{Z'} \vct{x}\|_2^2 \ge m/52) \le e^{-\gamma m}
\end{equation}
for some numerical constant $\gamma > 0$ with the proviso that $m \ge
cn$ and $c$ is sufficiently large. We omit the details. To conclude, 
\eqref{eq:bernstein} and \eqref{eq:chi-sq} give that with probability
at least $1 - 3e^{-\gamma m}$, 
\[
\|\vct{y}'\|_2^2 = \frac{1}{m^2} \|\mtx{Z'} \vct{c}\|^2_2 \le 
(1+\delta_0)/52 < 1/10.
\]
This concludes the proof. 
\end{proof}

\subsection{Proof of Theorem \ref{teo:main}} 

The proof of Theorem \ref{teo:main} is now a consequence of the
corollary below.

\begin{corollary}
\label{cor:core}
With probability at least $1 - O(e^{-\gamma m})$, for all $\vct{x}_0
\in \R^n$, there exists $\mtx{Y}$ obeying the conditions of Lemma
\ref{lem:inexact}. In addition, one can take $\mtx{Y} =
\cA^*(\vct{\lambda})$ with $\|\vct{\lambda}\|_\infty \le 7/m$.
\end{corollary}
The reason why this corollary holds is straightforward: Lemma
\ref{lem:core} holds true for exponentially points and a sort of
continuity argument allows to extend it to all points. Again it
suffices to establish the property for unit-normed vectors.

\begin{proof}
  Let $\cN_\epsilon$ be an $\epsilon$-net for the unit sphere with
  cardinality obeying $|\cN_\epsilon|\leq (1+ 2/\epsilon)^n$ by
  Lemma 2 in \cite{VershyninRMT}.\footnote{For any unit-normed vector
    $\vct{x}$, there is $\vct{x}_0 \in \cN_\epsilon$ with
    $\|\vct{x}_0\|_2 = 1$ and $\|\vct{x} - \vct{x}_0\|_2 \le
    \epsilon$, where $\epsilon > 0$.} If $c_0$ is sufficiently large,
  Lemma \ref{lem:core} implies that with probability at least
  $1-O(e^{-\gamma m}(1+ 2/\epsilon)^n)\geq 1-O(e^{-\gamma' m})$,
  for all $\vct{x}_0 \in \cN_\epsilon$, there exists $\mtx{Y} =
  \cA^*(\vct{\lambda})$ obeying
\begin{subequations}\label{eq:Y}
\begin{align}
\|\mtx{Y}_{T_0^\perp} + 1.7 \mtx{I}_{T_0^\perp}\| & \le 0.1\label{eq:YTp}\\
\|\mtx{Y}_{T_0}\|_F & \le 0.15\label{eq:YT}
\end{align}
\end{subequations}
and $\|\vct{\lambda}\|_\infty \le 7/m$ (we wrote $T_0$ in place of $T$
for convenience). Note that this gives 
\begin{equation*}
 \label{eq:Yopnorm}
 \|\mtx{Y}\|\leq  \|\mtx{Y}_{T_0}\| + \|\mtx{Y}_{T_0^\perp}\| \le 0.15 + 1.8 < 2.
\end{equation*}

Consider now an arbitrary unit-normed vector $\vct{x}$ and let
$\vct{x}_0 \in \cN_\epsilon$ be any element such that $\|\vct{x} -
\vct{x}_0\|_2\leq \epsilon$. Set $\mtx{\Delta}
=\mtx{x}\mtx{x}^T-\vct{x}_0 \vct{x}_0^T$, which obeys
 \begin{equation*}
 \label{ineq:diff}
 \|\mtx{\Delta}\|_F \leq \|(\vct{x}(\vct{x} -\vct{x}_0)^T)\|_F +
 \|(\vct{x}-\vct{x}_0) \vct{x}_0^T\|_F =  \|\vct{x}\|_2\|\vct{x}-\vct{x}_0\|_2 + \|\vct{x} - \vct{x}_0\|_2 
 \|\vct{x}_0\|_2 \leq 2\epsilon.
\end{equation*}
Suppose $\mtx{Y}$ is as in \eqref{eq:Y} and let $T$ be $\{\mtx{X} =
\vct{y} \vct{x}^T + \vct{x} \vct{y}^T : \vct{y} \in \R^n\}$. We have
\[
  \mtx{Y}_{T^\perp} + 1.7 \mtx{I}_{T^\perp}  = (\mtx{I} -
  \vct{x}\vct{x}^T) \mtx{Y} (\mtx{I} - \vct{x}\vct{x}^T) + 1.7(\mtx{I}
  -  \vct{x} \vct{x}^T) 
   = \mtx{Y}_{T_0^\perp} + 1.7 \mtx{I}_{T_0^\perp} - \mtx{R},
\]
where 
\[
\mtx{R} = \mtx{\Delta}\mtx{Y}(\mtx{I}-\vct{x}_0\vct{x}_0^T) +
(\mtx{I}-\vct{x}_0\vct{x}_0^T)\mtx{Y}\mtx{\Delta} -
\mtx{\Delta}\mtx{Y}\mtx{\Delta} + 1.7 \mtx{\Delta}. 
\] 
Since 
\[
\|\mtx{R}\|_2 \le 2 \|\mtx{Y}\| \|\mtx{\Delta}\| \|\mtx{I} - \vct{x}_0
\vct{x}_0^T\| + \|\mtx{Y}\| \|\mtx{\Delta}\|^2 + 1.7 \|\mtx{\Delta}\| \le
11.4 \epsilon + 8 \epsilon^2, 
\]
we see that the first condition of Lemma \ref{lem:inexact} holds
whenever $\epsilon$ is small enough. For the first condition, 
\[
\mtx{Y}_T = \vct{x}\vct{x}^T \mtx{Y} +(\mtx{I}-\vct{x}\vct{x}^T) 
\mtx{Y} \vct{x}\vct{x}^T = \mtx{Y}_{T_0} + \mtx{R},
\]
where 
\[
\mtx{R} = \mtx{\Delta}\mtx{Y}(\mtx{I}-\vct{x}_0\vct{x}_0^T) +
(\mtx{I}-\vct{x}_0\vct{x}_0^T) \mtx{\Delta}\mtx{Y} -
\mtx{\Delta}\mtx{Y}\mtx{\Delta}.
\] 
Since $\mtx{\Delta}$ has rank at most 2, $\rank(\mtx{R}) \le 2$, and 
\[
\|\mtx{R}\|_F \le \sqrt{2} \|\mtx{R}\| \le \sqrt{2}\Bigl(2\|\mtx{Y}\|
\|\mtx{\Delta}\| \|\mtx{I} - \vct{x}_0 \vct{x}_0^T\| + \|\mtx{Y}\|
\|\mtx{\Delta}\|^2\Bigr) \le 8\sqrt{2} (\epsilon + \epsilon^2).
\]
Choosing $\epsilon$ sufficiently small concludes the proof of the
corollary. 
\end{proof}

\subsection{Stability} 
\label{sec:stability} 

To see why the stability result \eqref{eq:stability} is optimal,
suppose without loss of generality that $\|\vct{x}_0\|_2 = 1$.
Further, imagine that we are informed that $\|\vct{w}\|_1 \le \delta
m$ for some known $\delta$. Since $\|\cA(\vct{x}_0\vct{x}_0^*)\|_1
\approx m$, it would not be possible to distinguish between solutions
of the form $(1+\lambda) \vct{x}_0 \vct{x}_0^*$ for $\max(0,1-\delta)
\lesssim 1 + \lambda \lesssim 1 + \delta$. Hence, the error in the
Frobenius norm may be as large as $\delta \|\vct{x}_0 \vct{x}_0^*\|_F
= \delta$, which is what the theorem gives.

We now turn to the proof of Theorem \ref{teo:stability}. We do not
need to show the second part as the perturbation argument is exactly
the same as in \cite[Theorem 1.2]{PhaseLift2}. The argument for the
first part follows that of the earlier Lemma \ref{lem:inexact}, and
makes use of the existence of a dual certificate $\mtx{Y} =
\mathcal{A}^*(\vct{\lambda})$ obeying the conditions of this lemma.

Set $\hat{\mtx{X}} = \vct{x}_0 \vct{x}_0^* + \mtx{H}$.  Since
$\hat{\mtx{X}}$ is feasible, $\mtx{H}_{T^\perp} \succeq \mtx{0}$ and
$\langle \mtx{H}, \mtx{Y} \rangle = \<\mathcal{A}(\mtx{H}),
\vct{\lambda}\>$.  First,
\begin{align*}
  \langle \mtx{H}_T, \mtx{Y}_T \rangle
  & = \<\mathcal{A}(\mtx{H}), \vct{\lambda}\> - \langle \mtx{H}_{T^\perp}, \mtx{Y}_{T^\perp}\rangle \\
  & \ge - \|\vct{\lambda}\|_\infty \|\cA(\mtx{H})\|_1 +
  \<\mtx{H}_{T^\perp}, \mtx{I}_{T^\perp}\> = -\frac{7}{m}
  \|\cA(\mtx{H})\|_1 + \tr(\mtx{H}_{T^\perp}).
\end{align*}
Second, we also have 
\[
\tr(\mtx{H}_{T^\perp})\geq {1\over
  {(1+1/8)m}}\|\cA(\mtx{H}_{T^\perp})\|_1 \geq  {1\over
  {(1+1/8)m}}(\|\cA(\mtx{H}_{T})\|_1 - \|\cA(\mtx{H})\|_1), 
\]
which by the same argument as before, yields
\[
\tr(\mtx{H}_{T^\perp}) \geq {0.73 \over {\sqrt{2}}}\|\mtx{H}_T\|_F -
{1\over {(1+1/8)m}}\|\cA(\mtx{H})\|_1.
\]
Since $|\langle \mtx{H}_T, \mtx{Y}_T \rangle| \leq \frac12
\|\mtx{H}_{T}\|_F$, we have established that\footnote{The careful
  reader will note that we can get a far better constant by observing
  that the proof of Theorem \ref{teo:main} also yields $\|\mtx{Y}_T\|_F \le
  1/4$. Hence, we have $\|\mtx{H}_T\|_F \le 4(8/9 + 7) {\|\cA(\mtx{H})\|_1}/{m}$.}
\[
\Bigl(\frac{0.73}{\sqrt{2}} - \frac{1}{2}\Bigr) \|\mtx{H}_T\|_F \le
  \Bigl(\frac{8}{9} + 7\Bigr)
    \frac{\|\cA(\mtx{H})\|_1}{m}. 
\]
Also, since $\mtx{H}_{T^\perp}$ is positive semidefinite 
\[
\|\mtx{H}_{T^\perp}\|_F \le \tr(\mtx{H}_{T^\perp}) \le \frac12 \|\mtx{H}_T\|_F +
\frac{7}{m} \|\cA(\mtx{H})\|_1
\]
so that 
\[
\|\mtx{H}\|_F \le C'_0 \frac{\|\cA(\mtx{H})\|_1}{m} \le 2C'_0 \frac{\|\vct{w}\|_1}{m} . 
\]
To see why the second inequality is true, observe that
\[
\|\vct{b} - \cA(\vct{x}_0 \vct{x}_0^* + \mtx{H})\|_1 = \|\vct{w} -
\cA(\mtx{H})\|_1 \le \|\vct{b} - \cA(\vct{x}_0 \vct{x}_0^*)\|_1 =
\|\vct{w}\|_1,
\]
which gives $\|\cA(\mtx{H})\|_1 \le 2 \|\vct{w}\|_1$ by the triangle
inequality. The proof is complete.

\small
\subsection*{Acknowledgements}
E.~C.~is partially supported by AFOSR under grant FA9550-09-1-0643 and
by ONR under grant N00014-09-1-0258.  This work was partially
presented at the University of California at Berkeley in January 2012,
and at the University of British Columbia in February 2012.

\bibliographystyle{plain} 
\bibliography{ImprovedPL}

\end{document}